\documentclass[journal,twocolumn]{IEEEtran}
\usepackage{tikz,xcolor,hyperref}

\definecolor{lime}{HTML}{A6CE39}
\DeclareRobustCommand{\orcidicon}{
	\begin{tikzpicture}
		\draw[lime, fill=lime] (0,0)
		circle[radius=0.16]
		node[white]{{\fontfamily{qag}\selectfont \tiny \.{I}D}}; 
	\end{tikzpicture}
	\hspace{-2mm}
}
\foreach \x in {A, ..., Z}{%
	\expandafter\xdef\csname orcid\x\endcsname{\noexpand\href{https://orcid.org/\csname orcidauthor\x\endcsname}{\noexpand\orcidicon}}
}

\usepackage{amsmath,amsfonts}
\usepackage{algorithmic}
\usepackage{algorithm}
\usepackage{array}
\usepackage[caption=false,font=normalsize,labelfont=sf,textfont=sf]{subfig}
\usepackage{textcomp}
\usepackage{stfloats}
\usepackage{url}
\usepackage{verbatim}
\usepackage{bm}
\usepackage{amsthm}
\usepackage{subfloat}
\usepackage{amsthm}
\usepackage{booktabs}
\usepackage{amssymb}
\newtheorem{theorem}{Theorem}

\usepackage{graphicx}
\usepackage{cite}
\usepackage{algorithm, algorithmic}

\usepackage{hyperref}
\hypersetup{colorlinks,
	linkcolor=red,
	citecolor=blue}
\begin{document}

\title{Precoder Design for User-Centric Network Massive
	MIMO: A Symplectic Optimization Approach}

\author{
	Pengxu~Lin\hspace{-1.5mm}\orcidA{},~\IEEEmembership{Student Member,~IEEE,}
	An-An~Lu\hspace{-1.5mm}\orcidB{},~\IEEEmembership{Member,~IEEE,} 
	and Xiqi~Gao\hspace{-1.5mm}\orcidC{},~\IEEEmembership{Fellow,~IEEE}
	
	\thanks{This work was supported by the National Natural Science Foundation of China under Grants 62394294, 62371125 and
62394290, the Jiangsu Province Major Science and Technology Project under Grant BG2024005, the Fundamental Research
Funds for the Central Universities under Grant 2242022k60007, the Key R$\&$D Plan of Jiangsu Province under Grant BE2022067,
and the Huawei Cooperation Project.
An earlier version of this paper will be presented in part at the IEEE Vehicular Technology Conference (VTC) 2025.}
	
	\thanks{Pengxu Lin, An-An Lu, and Xiqi Gao are with the National Mobile Communications Research Laboratory, Southeast University, Nanjing 210096, China, and also with the Purple Mountain Laboratories, Nanjing 211100, China (e-mail: 220230819@seu.edu.cn, aalu@seu.edu.cn, xqgao@seu.edu.cn).}}
%
%

\maketitle
\begin{abstract}
In this paper, we utilize symplectic optimization to design a precoder for user-centric network (UCN) massive multiple-input multiple-output (MIMO) systems, where a subset of base stations (BSs) serves each user terminal (UT) instead of using all BSs. In UCN massive MIMO systems, the dimension of the precoders is reduced compared to conventional network massive MIMO. It simplifies the implementation of precoders in practical systems. However, the matrix inversion in traditional linear precoders still requires high computational complexity. To avoid the matrix inversion, we employ the symplectic optimization framework, where optimization problems are solved based on dissipative Hamiltonian dynamical systems. To better fit symplectic optimization, we transform the received model into the real field and reformulate the weighted sum-rate (WSR) maximization problem. The objective function of the optimization problem is viewed as the potential energy of the dynamical system.  
Due to energy dissipation, the continuous dynamical system always converges to a state with minimal potential energy.
By discretizing the continuous system while preserving the symplectic structure, we obtain an iterative method for the precoder design.
The complexity analysis of the proposed symplectic method is also provided to show its high computational efficiency.
Simulation results demonstrate that the proposed precoder design based on symplectic optimization outperforms the weighted minimum mean-square error (WMMSE) precoder in the UCN massive MIMO system.
\end{abstract}


\section{Introduction} 
\IEEEPARstart{M}{assive} multiple-input multiple-output (MIMO) has played an essential role in the fifth-generation (5G) wireless communications and receives increasing attention for the sixth-generation (6G) wireless communications\cite{jiang2021road,wang2023road}. 
Massive MIMO has been recognized as a key technique for significantly enhancing spectral and energy efficiency by equipping large-scale antenna arrays at the base stations (BSs)\cite{lu2014overview,larsson2014massive}.
Cellular massive MIMO systems are the most widely adopted paradigm \cite{huang2011mimo}. The cellular system contains many cellular cells, and each cell contains a macro BS, which simultaneously serves many user terminals (UTs) in the cell on the same frequency resource.
The advantage of cellular massive MIMO systems in offering significant improvements in both spectral and energy efficiency has been verified by various works \cite{you2022energy,1678166,5595728}.
Nonetheless, cell-edge UTs experience significantly performance degradation because the received signal from its serving BS is weak and the inter-cell interference from neighboring BSs is strong \cite{rahman2010enhancing}.
Moreover, the system complexity becomes a bottleneck as the number of antennas at the BS increases. 

To improve the rate performance of cell-edge UTs, the cell-free massive MIMO system \cite{elhoushy2021cell,zhang2019cell}, in which a great amount of distributed access points (APs) coordinate transmission by connecting to a central processing unit (CPU) \cite{obakhena2021application}, is proposed.
Each UT in the cell-free massive MIMO system is served by all the APs and hence the notion of cell-edge disappears.
In addition, APs are equipped with far fewer antennas and are deployed more densely throughout the network.
However, as the number of access points (APs) and UTs increases, it becomes increasingly complex and costly to implement\cite{interdonato2020local}.
At the same time, network massive MIMO has also been investigated to improve the performance of UTs in the edge of a cell through joint transmission\cite{venkatesan2007network}.
In a network massive MIMO system, multiple BSs having a large antenna array share their channel state information (CSI). 
Moreover, by serving each UT with all BSs, it enhances the performance of UTs, and also reduces the possibility unnecessary handovers and link outages\cite{lee2012network}.
A smooth migration from cellular to network MIMO system is feasible while maintaining service continuity. However, the complexity is also an issue for network massive MIMO.

To reduce the complexity, the user-centric rule has been introduced to the cell-free and network massive MIMO systems \cite{ammar2021user,demir2021foundations,sun2025precoder}.
Specifically, each UT is served by a few BSs with the best channel conditions \cite{sun2025precoder}.
It reduces the dimension of the precoder and the system complexity compared to conventional network massive MIMO systems. 
Moreover, it also improves the performance of cell-edge UTs \cite{chen2023user}.
In this paper, we consider the user-centric network (UCN) massive MIMO system as in \cite{sun2025precoder}.
To obtain a satisfied performance for UCN massive MIMO systems, it is important to find an efficient precoder design approach to mitigate the inter-user interference \cite{albreem2021overview}. 

In traditional cellular massive MIMO system, various kinds of precoders have been proposed. 
Nonlinear precoders, such as the dirty paper coding (DPC) \cite{lee2007high}, achieve the capacity region, but with prohibitive complexity. 
The linear the signal-to-leakage-plus-noise ratio (SLNR)\cite{sadek2007leakage} and regularized zero forcing (RZF) \cite{peel2005vector}  precoders are well known for their implementation simplicity.
However, only sub-optimal performance is achieved since it does not directly maximize the sum-rate.
The weighted minimum mean-squared error (WMMSE) precoder \cite{christensen2008weighted, shi2011iteratively} solves the weighted sum-rate (WSR) maximization problem and thus outperforms the RZF precoder.
Nevertheless, it involves the matrix inversion, which leads to intensive computation when applied to the UCN massive MIMO system since there are large number of users and transmit antennas\cite{shi2023robust}. 
The introduction of the UCN effectively reduces the computational complexity, but the complexity is still high as the dimension of the matrix inversion is still large \cite{sun2025precoder}.
To avoid the matrix inversion, the gradient descent method can be used to solve the WSR problem directly, but it converges slowly. 
To improve the convergence behavior, the accelerated gradient method can be employed\cite{ji2009accelerated}. 
However, the accelerated gradient method has not emerged a theoretical framework until the appearance of symplectic optimization \cite{betancourt2018symplectic}. 

Symplectic optimization solves an optimization problem with a dissipative Hamiltonian dynamical system.  The objective function of the optimization problem is viewed as the potential energy in the dynamical system \cite{betancourt2018symplectic} and a kinetic energy term is added. 
It reaches the minimal value of the potential function when the energy is dissipated, and is expected to achieve better convergence performance \cite{ghirardelli2023optimization}. 
Symplectic optimization has the potential to run away from local optimal points and achieve fast convergence.
It is first introduced to wireless communications in \cite{zhang2025cross} to solve an unconstrained quadratic sub-problem in single cell massive MIMO to reduce complexity.
In \cite{zhang2025symplecticoptimizationcrosssubcarrier}, symplectic optimization is used to design cross sub-carrier precoder with channel smoothing in the single cell massive MIMO system, achieving improved performance and reduced computational complexity.
Symplectic optimization has better performance when applied to large-scale data analysis \cite{jordan2018dynamical}.


In this paper, we employ symplectic optimization to propose the precoder design for the UCN massive MIMO system. 
The proposed precoder has also been extended to robust precoder design to enhance the sum-rate performance in high-mobility scenarios\cite{Lin2025}.
We formulate the WSR maximization problem with a set of power constraints due to the limited transmit power of each BS. 
Meanwhile, we transform the received model into the real field to better fit the symplectic optimization method.
Symplectic optimization is applied directly to solve the constrained WSR maximization problem, and  an efficient numerical scheme is proposed.
The large dimensional matrix inversion is avoided in the symplectic optimization method. 
To demonstrate the high efficiency of the proposed symplectic optimization method, its computational complexity is analyzed.
Simulation results show that the symplectic optimization precoder design method achieves satisfied performance in UCN massive MIMO systems.

The main contribution of this work is provided as follows.

1. We introduce the advanced symplectic optimization for the precoder design in the UCN massive MIMO system. 

2. We derive an iterative symplectic method for the precoder design. The proposed method converges faster than the accelerated gradient method, and outperforms the WMMSE precoder with lower complexity. 

The rest of this paper is organized as follows.
In Section \ref{section2}, the system model and problem formulation are provided.
The symplectic optimization approach for precoder design is provided in Section \ref{section3}.
Simulation results are presented in Section \ref{section4}.
Section \ref{section5} draws the conclusion. Proofs of the theorems are given in Appendices. 

\textit{Notations}: Boldface lowercase and uppercase letters represent the column vectors and matrices, respectively. 
The superscripts $(\cdot)^*$, $(\cdot)^T$, and $(\cdot)^H$ denote the conjugate, transpose, and conjugate-transpose, respectively. $\Re\{A\}$ and $\Im\{A\}$ mean the real part and imaginary part of \(A\). $\mathrm{diag}(\mathbf{a})$ represents the diagonal matrix with $\mathbf{a}$ along its main diagonal. Similarly, $\mathbf{D}=\mathrm{vecdiag}\{\mathbf{a}_1,\cdots,\mathbf{a}_L\}$ denotes the vector diagonal matrix with $\mathbf{a}_1,\cdots,\mathbf{a}_L$ on the main diagonal.

\section{System Model and Problem Formulation}
\label{section2}

\subsection{System model}
Consider a UCN massive MIMO system consisting of $K$ UTs and $B$ BSs. 
Each BS in the system is equipped with a uniform planar array (UPA) having $M_t=M_v \times M_h$ antennas, and each UT has $M_r$ antennas. 
The system works in the time division duplexing (TDD) mode.
We assume the BSs are synchronized and interconnected through backhaul links. 
Thus, the BSs can perform joint coherent transmission. 
Let $\mathcal{S}_B=\{1,2,\dots, B\}$ denote the set of all BSs, and $\mathcal{S}_U=\{1,2,\dots, K\}$ denote the set of all UTs.
For a specific user \( k \), it is served by the set $\mathcal{B}_{k}=\{k_1,k_2,\dots,k_{B_{k}}\}$, which is a subset of $\mathcal{S}_B$. Similarly, for a specific BS \( l \), the set of users it serves can be denoted as $\mathcal{U}_l=\{1,2,\dots,l_{U_l}\}$. This user-centric rule enables each UT to receive signals from the base station with the best channel quality without relying on the traditional cell concept.
Fig.~\ref{System} illustrates the UCN massive MIMO system, where, for clarity, only three UTs and their corresponding serving clusters are depicted.

\begin{figure}[htbp]
	\centerline{\includegraphics[width=1\columnwidth]{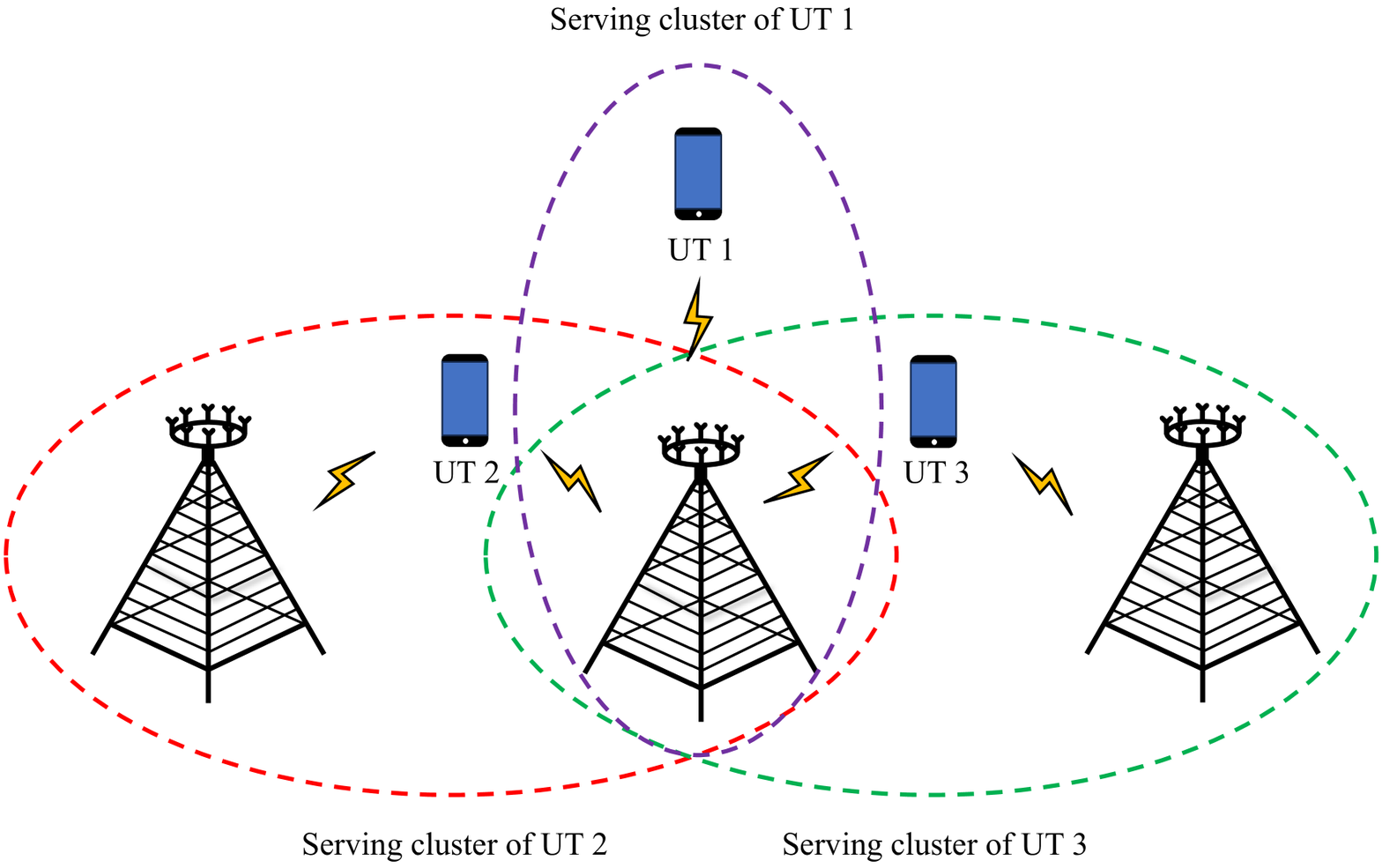}}
	\caption{The UCN massive MIMO system.}
	\label{System}
\end{figure}

Let ${x}_k$ denote the symbol transmitted to the \(k\)-th UT with $\mathbb{E}\{x_kx_k^*\}=1, k\in \mathcal{S}_U$, and $\mathbb{E}\{x_kx_j^*\}=0,j\ne k,j\in \mathcal{S}_U$. The channel vector from the \(l\)-th base station to the \(k\)-th user is denoted as $\mathbf{h}_{l,k}\in\mathbb{C}^{M_t\times M_r}$. Let $\mathbf{p}_{l,k}\in\mathbb{C}^{M_t\times1}$ denote the precoder vector designed for the transmission from BS v\(_l\) to UT \(k\), $l\in \mathcal{B}_k$. The received signal at UT \(k\) is given by
\begin{align}
	y_k=&\sum_{l\in \mathcal{B}_k}\mathbf{h}_{l,k}^H\mathbf{p}_{l,k}x_k+\sum_{l\in \mathcal{B}_k}\sum_{t\in \mathcal{U}_l,t\ne k}\mathbf{h}_{l,k}^H\mathbf{p}_{l,t}x_t \nonumber\\ 
	&+\sum_{m\ne \mathcal{B}_k}\sum_{t\in \mathcal{U}_m}\mathbf{h}_{m,k}^H\mathbf{p}_{m,t}x_t+z_k \label{yk}
\end{align}
where \(z_k\) is the complex Gaussian noise with distribution $\mathcal{CN}(0,\sigma_z^2)$. Let \(z_k'\) denote the interference plus noise of UT \(k\) defined as 
\begin{align}
	z_k'&=\sum_{l\in \mathcal{B}_k}\sum_{t\in \mathcal{U}_l,t\ne k}\mathbf{h}_{l,k}^H\mathbf{p}_{l,t}x_t+\sum_{m\ne \mathcal{B}_k}\sum_{t\in \mathcal{U}_m}\mathbf{h}_{m,k}^H\mathbf{p}_{m,t}x_t+z_k\nonumber\\
	&=\sum_{t\ne k}\sum_{m\in \mathcal{B}_t}\mathbf{h}_{m,k}^H\mathbf{p}_{m,t}x_t+z_k
\end{align}
whose covariance is given by
\begin{align}
	r_k=\sum_{t\ne k}\Big(\sum_{m\in \mathcal{B}_t}\mathbf{h}_{m,k}^H\mathbf{p}_{m,t}\Big)\Big(\sum_{m\in \mathcal{B}_t}\mathbf{h}_{m,k}^H\mathbf{p}_{m,t}\Big)^*+\sigma_z^2. \label{rk1}
\end{align}
The rate of UT \(k\) is easy to obtain as
\begin{align}
	\mathcal{R}_k=\log_2\Big(1+  r_{k}^{-1}
	(\sum_{l\in\mathcal{B}_k}\mathbf{p}_{l,k}^H\mathbf{h}_{l,k})(\sum_{l\in\mathcal{B}_k}\mathbf{h}_{l,k}^H\mathbf{p}_{l,k})\Big).
\end{align}

\subsection{Problem Formulation}

For simplicity, let $\mathbf{p}_{l}=[\mathbf{p}_{l,1}^T,\mathbf{p}_{l,2}^T,\dots,\mathbf{p}_{l,K}^T]^T$ and $\mathbf{p}=[\mathbf{p}_1^T,\mathbf{p}_2^T,\dots,\mathbf{p}_L^T]^T$.

Since each BS has its individual power constraint in the UCN massive MIMO system, we employ the WSR-maximization precoder design method as
\begin{equation}
	\begin{split}
		&\mathrm{arg} \max_{\mathbf{p}} f(\mathbf{p})=\sum_{k\in \mathcal{S}_U}w_k\mathcal{R}_k \\ 
		&\mathrm{s.t.}\sum_{k\in\mathcal{K}}\mathbf{p}_{l,k}^{H}\mathbf{p}_{l,k} \leq\rho_l\quad\forall l\in\mathcal{L}\label{WSR}
	\end{split}
\end{equation}
where $w_k$ is the weight and \(\rho_l\) denotes the power constraint of \(l\)-th BS.

The WSR-problem \eqref{WSR} can be solved directly by the iterative WMMSE method. 
For simplicity, we denote
\begin{align}
	A_{k,t} = \sum_{m\in \mathcal{B}_t}\mathbf{h}_{m,k}^H\mathbf{p}_{m,t}
\end{align} 
and 
\begin{align}
	\check{r}_k&=\sum_{j \in \mathcal{S}_U }\Big(\sum_{m\in \mathcal{B}_t}\mathbf{h}_{m,k}^H\mathbf{p}_{m,t}\Big)\Big(\sum_{m\in \mathcal{B}_t}\mathbf{h}_{m,k}^H\mathbf{p}_{m,t}\Big)^*+\sigma_z^2\nonumber\\
	&=\sum_{j \in \mathcal{S}_U }A_{k,t}A_{k,t}^* +\sigma_z^2.
\end{align}
From the model in \eqref{yk}, the WMMSE precoder is derived as
\begin{align}
	\mathbf{p}_{l,k}^\mathrm{WMMSE} & = (\sum_{j \in \mathcal{S}_U } w_j \mathbf{h}_{l,j} u_j^*\mathrm{W}_ju_j \mathbf{h}_{l,j}^H + \lambda _k\mathbf{I}_{M_t} )^{-1}\nonumber\\
	&(w_k\mathbf{h}_{l,k}u_k^*\mathrm{W}_k-\sum_{j \in \mathcal{S}_U }w_j\mathbf{h}_{l,j}u_j^*\mathrm{W}_ju_j\sum_{m \neq l}\mathbf{h}_{m,j}^H\mathbf{p}_{m,k})  \label{WMMSE}
\end{align}
where \(u_k\) denotes the receive coefficient for UT \(k\) and is expressed as
\begin{align}
	u_k = A_{k,t}^*\check{r}_k^{-1}
\end{align}
and $\mathrm{W}_k$ is the MMSE weight, which can be obtained by
\begin{align}
	\mathrm{W}_k^* = 1 + A_{k,k}^*r_k^{-1}A_{k,k}.
\end{align}

The \(\lambda_k\) used in \eqref{WMMSE} denotes the Lagrange multiplier, which can be computed using a bisection method.
As shown in \eqref{WMMSE}, the WMMSE method involves the inversion of an \(M_t\times M_t\) matrix, which results in high computational complexity.
This significantly limits its use in practical UCN massive MIMO system.


To avoid the matrix inversion, conventional gradient descent (GD) methods can be used to solve the WSR problem directly.  Define $g(\hat{\mathbf{p}})=-f(\hat{\mathbf{p}})$.
The conventional update formula is given by
\begin{align}
	\mathbf{p}_{n+1}^{\mathrm{GD}} = \mathbf{p}_{n}^{\mathrm{GD}} - \alpha \nabla g(\mathbf{p}_{n}^{\mathrm{GD}})
\end{align}
where $\nabla g(\mathbf{p}_{n}^{\mathrm{GD}})$ denotes the gradient and \(\alpha\) is the step size.

Due to the slow convergence of the conventional GD method, an accelerated variant known as Nesterov accelerated gradient descent (NAGD) \cite{nesterov1983method} has been widely adopted to enhance the convergence performance\cite{liu2022convergence,sutskever2013importance}.
The NAGD method introduces a momentum term that can significantly speed up the iterative process and improve convergence\cite{nesterov1983method}. 
It has been applied to the precoder design in cell-free-assisted LEO satellite communications with user-centric rule \cite{miao2025serving}. 
The update formula of the NAGD method is given by
\begin{align}
	&\mathbf{q}_n = \mathbf{p}_{n}^{\mathrm{NAGD}} + \mu(\mathbf{p}_{n}^{\mathrm{NAGD}} - \mathbf{p}_{n-1}^{\mathrm{NAGD}})\\
	&\mathbf{p}_{n+1}^{\mathrm{NAGD}} = \mathbf{q}_n - \alpha \nabla g(\mathbf{p}_{n}^{\mathrm{NAGD}})
\end{align}
where \(\mu\) is hyperparameter, \(\alpha\) is the step size, and \(\mathbf{q}\) denotes the momentum. 

In the gradient decent method, how to choose the step size is crucial to ensure convergence. 
Line search methods are widely used strategies in gradient decent algorithms to enhance convergence behavior \cite{nocedal1999numerical}.
The Armijo condition is a criterion commonly used in line search methods to determine whether a candidate step size an produce a sufficient decrease in the objective function \cite{nocedal1999numerical}.
However, the line search method typically requires multiple calculations of the objective function in each iteration to determine a suitable step size, which can significantly increase the computational complexity\cite{sun2006optimization}.

Symplectic optimization is an advanced method proposed in recent years to solve mathematical optimization problems. 
It relates the optimization problem with a dissipative dynamical system, where the potential energy
is the objective function and a kinetic energy term is also included. Due to energy dissipation, the continuous dissipative dynamical system always converges to a minimal potential energy, which is also a minimal value of the optimization problem\cite{harier2000geometric}. 
By using discretization that keeps the symplectic structure, symplectic optimization method that preserve the properties of the original continuous dynamical system is obtained. 
This means that algorithms based on symplectic optimization run faster compared to the gradient descent method. 
Additionally, the symplectic optimization method is more likely to escape local optimal point due to the kinetic energy. 
Thus, it might outperform the WMMSE algorithm which might be trapped in a local optimal point. 

To reduce the computational complexity and improve the system performance, we apply symplectic optimization to the precoder design for the considered UCN massive MIMO system in the following section.

\section{Symplectic Optimization based Precoder Design}
\label{section3}
In this section, we first transform the received model into the real field to better employ the symplectic optimization method.
Next, we introduce the augmented Hamiltonian dynamics from the augmented Lagrangian formulation and incorporate a dissipation mechanism to derive the precoder. 
Finally, we employ a numerical scheme to solve the dissipative Hamiltonian system and analyze the computational complexity of the proposed method.

\subsection{Problem Reformulation in Real Domain}
Symplectic optimization is formulated in the real domain \cite{harier2000geometric}, whereas the receive model in the previous section is defined in the complex domain. 
Therefore, to make the application of symplectic optimization rigorous, the considered problem is transformed into its real-valued representation \cite{liu2013comparisons}.

Define $\hat{\mathbf{u}}$ and $\check{\mathbf{H}}$ as 
\begin{align}
	\hat{\mathbf{u}}&=\begin{bmatrix}
		\Re\{u\} \\
		\Im\{u\}
	\end{bmatrix}\\
	\check{\mathbf{H}}&=\begin{bmatrix}
		\Re\{h\} & -\Im\{h\}\\
		\Im\{h\} & \Re\{h\}
	\end{bmatrix}.
\end{align}
Then, $\hat{\mathbf{y}}_k$, $\hat{\mathbf{p}}_{l,k}$ and $\hat{\mathbf{z}}_k$ are defined similarly. 
The received signal model (\ref{yk}) and the rate of UT \(k\) is rewritten as 
\begin{align}
	\hat{\mathbf{y}}_k&=\sum_{l\in \mathcal{B}_k}\check{\mathbf{H}}_{l,k}^T\hat{\mathbf{p}}_{l,k}x_k+\sum_{l\in \mathcal{B}_k}\sum_{t\in \mathcal{U}_l,t\ne k}\check{\mathbf{H}}_{l,k}^T\hat{\mathbf{p}}_{l,t}x_t\nonumber \\ 
	&~~+\sum_{m\ne \mathcal{B}_k}\sum_{t\in \mathcal{U}_m}\check{\mathbf{H}}_{m,k}^T\hat{\mathbf{p}}_{m,t}x_t+\hat{\mathbf{z}}_k  
\end{align}
and
\begin{align}
	\mathcal{R}_k=\log\Big(1+  r_{k}^{-1}
	(\sum_{l\in\mathcal{B}_k}\hat{\mathbf{p}}_{l,k}^T\check{\mathbf{H}}_{l,k})(\sum_{l\in\mathcal{B}_k}\check{\mathbf{H}}_{l,k}^T\hat{\mathbf{p}}_{l,k})\Big)
\end{align}
where 
\begin{align}
	r_{k}= \sum_{t\ne k}(\sum_{m\in\mathcal{B}_t}\hat{\mathbf{p}}_{m,t}^T\check{\mathbf{H}}_{m,k} )(\sum_{m\in\mathcal{B}_t}\check{\mathbf{H}}_{m,k}^T\hat{\mathbf{p}}_{m,t})+ \sigma_z^2\label{rk}.
\end{align}

Recall that $g(\hat{\mathbf{p}})=-f(\hat{\mathbf{p}})$. 
Define the constraint set $\mathcal{P}$ as $\mathcal{P} = \{\hat{\mathbf{p}}|\bm{\phi}(\hat{\mathbf{p}}) = \bm{\rho}\}$.
It also constrains the movement of \(\hat{\mathbf{p}}\).
We rewrite the optimization problem in \eqref{WSR}  as
\begin{align}
	\underset{\hat{\mathbf{p}}\in\mathcal{P}}{\operatorname*{\arg\min}} g(\hat{\mathbf{p}}).
\end{align}

\subsection{Augmented Lagrangian Dynamics}

Symplectic optimization relates optimization problems with dynamical systems. 
Lagrangian and Hamiltonian Systems are two widely used dynamical systems. 
To apply symplectic optimization, we first introduce augmented Lagrangian systems for the precoder design, where ``augmented'' is used since there exists constraints.

The Lagrangian dynamical system is constructed as follows. First, the objective function $g(\hat{\mathbf{p}})$ is viewed as its potential energy, and the precoder vector \(\hat{\mathbf{p}}\) is its position variable.  
Let \(\mathbf{M}\) be the mass matrix, the kinetic energy of the dynamical system is defined as $T(\dot{\hat{\mathbf{p}}})=\frac{1}{2}\dot{\hat{\mathbf{p}}}^T\mathbf{M}\dot{\hat{\mathbf{p}}}$. The Lagrangian is defined as the kinetic energy $T(\dot{\hat{\mathbf{p}}})$ minus the potential energy $g(\hat{\mathbf{p}})$, \textit{i.e.}, $T(\dot{\hat{\mathbf{p}}})-g(\hat{\mathbf{p}})$. When there exists the constraint, 
the augmented Lagrangian is defined with the Lagrange multiplier \(\bm\lambda\) as \cite{harier2000geometric}
\begin{align}
	L(\hat{\mathbf{p}},\dot{\hat{\mathbf{p}}})=T(\dot{\hat{\mathbf{p}}})-g(\hat{\mathbf{p}})-(\bm{\phi}(\hat{\mathbf{p}})-\bm{\rho})^T\bm{\lambda}. \label{Langrange}
\end{align}
The equation of the variational formulation is given as
\begin{align}
	\frac d{dt}(\frac{\partial L}{\partial\dot{\hat{\mathbf{p}}}})-\frac{\partial L}{\partial\hat{\mathbf{p}}}=\mathbf{0}.\label{motion}
\end{align}
Combining (\ref{Langrange}) with (\ref{motion}), we have the following second-order differential equation
\begin{align}
\mathbf{M}\ddot{\hat{\mathbf{p}}}+\nabla g(\hat{\mathbf{p}})+\mathbf{G}(\hat{\mathbf{p}})^T\bm{\lambda}=\mathbf{0}\label{chushi formulate}
\end{align}
where $\nabla g(\hat{\mathbf{p}})$ denotes the gradient of $g(\hat{\mathbf{p}})$, and $G(\hat{\mathbf{p}})$ is the Jacobian matrix of the vector-valued function $\bm\phi(\hat{\mathbf{p}})$ given as $G(\hat{\mathbf{p}})=\frac{\partial}{\partial\hat{\mathbf{p}}}\bm{\phi}(\hat{\mathbf{p}})$. 

Define $\mathbf{v}$ as the velocity coordinate of the position coordinate $\hat{\mathbf{p}}$,
equation \eqref{chushi formulate} and the constraint can be rewritten as a first-order differential equation of the form  
\begin{subequations}\label{first}
	\begin{align}
		&\dot{\hat{\mathbf{p}}}=\mathbf{v} \\ 
		&\mathbf{M}\dot{\mathbf{v}}=-\nabla g(\hat{\mathbf{p}})-\mathbf{G}(\hat{\mathbf{p}})^T\bm\lambda \\ 
		&\bm{\phi}(\hat{\mathbf{p}})= \bm{\rho}.
	\end{align}
\end{subequations}

\par

The augmented Lagrangian dynamical system is time-independent, which means its flow would oscillate around the minimal potential energy. To make the dynamical system converge to the minimal potential energy, we need to define dissipative Lagrangian systems. Furthermore, to obtain a practical algorithm, we have to discretize the continuous system. However, discretizing Lagrangian systems is often fragile\cite{betancourt2018symplectic} and lead to unstable algorithms. Thus, the dissipative augmented Hamiltonian system, which is the dual of the augmented Lagrangian system and a stable dissipative dynamical system \cite{ghirardelli2023optimization}, is introduced in the following.

\subsection{Dissipative Augmented Hamiltonian Dynamics}

The Hamiltonian dynamical system is constructed as follows. 
By defining the momentum coordinate $\hat{\mathbf{q}}$ as the Legendre dual of the position coordinate, we have
\begin{align}
	\hat{\mathbf{q}}=\frac{\partial L}{\partial \dot{\hat{\mathbf{p}}}}=\mathbf{M}\dot{\hat{\mathbf{p}}}.
\end{align}
The Hamiltonian  is defined as the kinetic energy plus potential energy, \textit{i.e.}, $T(\hat{\mathbf{q}}) + g(\hat{\mathbf{p}})$,
where $T(\hat{\mathbf{q}})=\frac{1}{2}\hat{\mathbf{q}}^T\mathbf{M}^{-1}\hat{\mathbf{q}}$.

Since the formula of $\nabla \mathbf{g}(\hat{\mathbf{p}})$ will be used subsequently, we present  its expression here. Recall that $g(\hat{\mathbf{p}})=-\sum_{k\in \mathcal{S}_U}\mathcal{R}_{k}$. From  
\begin{align}
	\frac{\partial\mathcal{R}_{k}}{\partial\hat{\mathbf{p}}_{l,k}}&=r_{k}^{-1}b_{k}\check{\mathbf{H}}_{l,k}(\sum_{l\in \mathcal{B}_k}\check{\mathbf{H}}_{l,k}^T\hat{\mathbf{p}}_{l,k})\\ 
	\frac{\partial\mathcal{R}_{t}}{\partial\hat{\mathbf{p}}_{l,k}}&=r_t^{-2}a_{t}b_{t}\check{\mathbf{H}}_{l,t}(\sum_{l\in \mathcal{B}_k}\check{\mathbf{H}}_{l,t}^T\hat{\mathbf{p}}_{l,k})
\end{align}
where  
\begin{align}
	&a_{k} = (\sum_{l\in \mathcal{B}_k}\hat{\mathbf{p}}_{l,k}^T\check{\mathbf{H}}_{l,k})(\sum_{l\in \mathcal{B}_k}\check{\mathbf{H}}_{l,k}^T\hat{\mathbf{p}}_{l,k})\label{ak}\\
	&b_{k} =(1+a_{k}r_k^{-1})^{-1}  \label{bk}
\end{align}
we have
\begin{align}
	\nabla \mathbf{g}(\hat{\mathbf{p}}_{l,k})&=
	-r_{k}^{-1}b_{k}\check{\mathbf{H}}_{l,k}(\sum_{l\in \mathcal{B}_k}\check{\mathbf{H}}_{l,k}^T\hat{\mathbf{p}}_{l,k}) \nonumber \\
	&~~+\sum_{t\neq k}^Kr_t^{-2}a_{t}b_{t}\check{\mathbf{H}}_{l,t}(\sum_{l\in \mathcal{B}_k}\check{\mathbf{H}}_{l,t}^T\hat{\mathbf{p}}_{l,k}). \label{Hp}
\end{align}
Then, we obtain $\nabla\mathbf{g}({\hat{\mathbf{p}}})=[\nabla\mathbf{g}(\hat{\mathbf{p}}_{1})^T,\cdots,\nabla\mathbf{g}(\hat{\mathbf{p}}_{L})^T]^T$.

Similarly to the augmented Lagrangian, define the augmented Hamiltonian as 
\begin{align}
	H(\hat{\mathbf{p}},\hat{\mathbf{q}})=T(\hat{\mathbf{q}}) + g(\hat{\mathbf{p}}) +(\bm{\phi}(\hat{\mathbf{p}})-\bm{\rho})^T\bm{\lambda}. \label{Ham}
\end{align}
Let \(H_\mathbf{p}\) and \(H_\mathbf{q}\) denote the partial gradients of the augmented Hamiltonian with respect to \(\hat{\mathbf{p}}\) and \(\hat{\mathbf{q}}\), respectively. 

In augmented Hamiltonian system, equation $\eqref{first}$ becomes 
\begin{subequations}\label{LDE}
	\begin{align}
		&\dot{\hat{\mathbf{p}}}=H_{\mathbf{q}} = \mathbf{M}^{-1}\hat{\mathbf{q}}\\ 
		&\dot{\hat{\mathbf{q}}}=-H_{\mathbf{p}}  = -\nabla \mathbf{g}(\hat{\mathbf{p}})-\mathbf{G}(\hat{\mathbf{p}})^T\bm{\lambda}\label{LDE2} \\
		&\bm{\phi}(\hat{\mathbf{p}})=\bm{\rho}. 
	\end{align}
\end{subequations}  
By differentiating $\bm{\phi}(\hat{\mathbf{p}})=\bm{\rho}$ in \eqref{LDE} with respect to
 $t$ once and twice, we obtain \cite{harier2000geometric} 
\begin{subequations}\label{Constrain}
	\begin{align}
		&\mathbf{G}(\hat{\mathbf{p}})\dot{\hat{\mathbf{p}}} =\mathbf{G}(\hat{\mathbf{p}}) \mathbf{M}^{-1}\hat{\mathbf{q}}=\mathbf{0} \label{26a} \\ 
		&\frac{\partial}{\partial\hat{\mathbf{p}}}\Big(\mathbf{G}(\hat{\mathbf{p}}) \mathbf{M}^{-1}\hat{\mathbf{q}}\Big)\dot{\hat{\mathbf{p}}}+\mathbf{G}(\hat{\mathbf{p}})\mathbf{M}^{-1}\dot{\hat{\mathbf{q}}}=0. \label{26b}
	\end{align}
\end{subequations}
By solving \eqref{26b} with \eqref{LDE2}, we can get the expression of $\bm{\lambda}$ in the following theorem.

\begin{theorem}\label{thm1}
	The expression of $\bm{\lambda}$ satisfying \eqref{26b} is given by
	\begin{align}
		{\bm{\lambda}}=\mathbf{C}^{-1}\Big(\mathbf{M}^{-1}\mathbf{Q}\hat{\mathbf{q}} -\mathbf{M}\mathbf{G}(\hat{\mathbf{p}})\mathbf{M}^{-1}\nabla\mathbf{g}(\hat{\mathbf{p}})\Big)\in \mathbb{R}^{L \times 1}\label{lambda}
	\end{align} 
	where $\mathbf{C}=\mathrm{diag}(\bm{\rho})$
	denotes the precoder vector power factor, and $\mathbf{Q}=\mathrm{vecdiag}\{\hat{\mathbf{q}}_1^T,\cdots,\hat{\mathbf{q}}_L^T\}$. 
\end{theorem} 
\begin{proof}
	See in Appendix \ref{fulu1}.
\end{proof}
Inserting the obtained expression ${\bm{\lambda}}$ from \eqref{lambda} into augmented Hamiltonian system \eqref{LDE} gives a differential equation on the  manifold $\mathcal{M}$ defined as \cite{harier2000geometric}
\begin{align}
\mathcal{M}	= \{(\hat{\mathbf{p}}, \hat{\mathbf{q}})\ |\ \bm{\phi}(\hat{\mathbf{p}})=\bm{\rho},\ \mathbf{G}(\hat{\mathbf{p}})\mathbf{M}^{-1}\hat{\mathbf{q}} = 0\}.\label{manifold} 
\end{align}
Define the configuration manifold as \(\mathcal{P} = \{\hat{\mathbf{p}}\ ;\  \bm{\phi}(\hat{\mathbf{p}})=\bm{\rho}\}\), which denotes the constrained position space.
For a fixed $\hat{\mathbf{p}} \in \mathcal{P}$, the Lagrangian \(	L(\hat{\mathbf{p}},\dot{\hat{\mathbf{p}}})\) is a function on the tangent space \(T_{p}\mathcal{P}\).
In this system, the tangent space \( T_{p}\mathcal{P} \) represents the set of all possible directions of motion at the point \( p \) on the manifold \( \mathcal{P} \).
Through the Legendre transformation, we can obtain the corresponding cotangent space \( T_p^*\mathcal{P} \) at the point \( \hat{\mathbf{p}}\) with the identification
\begin{align}
	T_p^*\mathcal{P} = \{ \hat{\mathbf{q}} | \hat{\mathbf{q}} = \mathbf{M} \dot{\hat{\mathbf{p}}} \  ; \  \dot{\hat{\mathbf{p}}} \in T_p\mathcal{P}\}.
\end{align}
Let \(T^* \mathcal{P}\) be the  cotangent bundle of \(\mathcal{P}\) defined as \(T^* \mathcal{P} = \{(\hat{\mathbf{p}} , \hat{\mathbf{q}})\} | \hat{\mathbf{p}} \in \mathcal{P}\ , \ \hat{\mathbf{q}} \in T_{p}^*\mathcal{P}\}\) .
The manifold \(\mathcal{M}\)  in \eqref{manifold} satisfies
\begin{align} 
	\mathcal{M} = T^*\mathcal{P}.
\end{align}
Therefore, \(\mathcal{M}\) is the cotangent bundle of \(\mathcal{P}\), and the augmented Hamiltonian system \eqref{LDE} can be viewed as a system on the cotangent bundle of \(\mathcal{P}\) \cite{harier2000geometric}. 
As a result, the Hamiltonian system is expressed within the framework of differential geometry on the manifold.

In addition, the symmetry \(T(\dot{\mathbf{p}}) = T(-\dot{\mathbf{p}})\) implies that system \eqref{Ham} is reversible.
By differentiate the Hamiltonian system of \eqref{Ham} with respect to $t$, we have
\begin{align}
-(\mathbf{M}^{-1}\hat{\mathbf{q}})^T \nabla \mathbf{g}(\hat{\mathbf{p}}) - (\mathbf{M}^{-1}\hat{\mathbf{q}})^T\mathbf{G}(\hat{\mathbf{p}})^T\bm{\lambda} + \nabla \mathbf{g}(\hat{\mathbf{p}})^T\mathbf{M}^{-1}\hat{\mathbf{q}}.
\end{align}
It is obviously that the first and last term cancel, and the second term equals to \({0}\) according to \eqref{26a}.
Thus, the Hamiltonian system \eqref{LDE} is preservation and remains constant along solutions of \eqref{LDE}.
Moreover, the flow produced by \eqref{LDE} is a symplectic transformation on \(\mathcal{M}\)\cite{harier2000geometric}.

As previously mentioned, the augmented Hamiltonian system is time-independent and conservative, which results in the system that potential is not easily minimized. To make the system converge to the minimal potential energy, dissipation is introduced, \textit{i.e.}, multiplying the momentum by a coefficient $\gamma$ to represent the momentum decreasing over time. \par
Let the dissipative augmented Hamiltonian be defined as   \cite{ghirardelli2023optimization, mclachlan2001conformal} 
\begin{align}
	\tilde{H}(\breve{\mathbf{p}},\breve{\mathbf{q}})   =e^{\gamma t}(H(\breve{\mathbf{p}},e^{-\gamma t}\breve{\mathbf{q}}))
\end{align}
where $(\breve{\mathbf{p}},\breve{\mathbf{q}})$ is defined as $(\breve{\mathbf{p}},\breve{\mathbf{q}})=(\hat{\mathbf{p}},e^{\gamma t}\hat{\mathbf{q}})$. The flow of the dissipative augmented Hamiltonian system are determined by the following differential equations  \cite{mclachlan2001conformal} 
\begin{subequations}
	\begin{align}
		\dot{\breve{\mathbf{p}}}&=\tilde{H}_{\breve{\mathbf{q}}}(\breve{\mathbf{p}},\breve{\mathbf{q}})=\mathbf{M}^{-1}\hat{\mathbf{q}}=\dot{\hat{\mathbf{p}}}\\
		\dot{\breve{\mathbf{q}}}&=-\tilde{H}_{\breve{\mathbf{p}}}(\breve{\mathbf{p}},\breve{\mathbf{q}})\nonumber\\
		&=-e^{\gamma t}(\nabla \mathbf{g}(\hat{\mathbf{p}})+\mathbf{G}(\hat{\mathbf{p}})^T\bm{\lambda}).
	\end{align}
\end{subequations}
From $\breve{\mathbf{q}}=e^{\gamma t}\hat{\mathbf{q}}$, we obtain $\dot{\breve{\mathbf{q}}}=e^{\gamma t}(\dot{\hat{\mathbf{q}}}+\gamma \hat{\mathbf{q}})$. Then,  
we have
\begin{subequations}\label{30}
	\begin{align}
		&\mathbf{M}\dot{\hat{\mathbf{p}}}=\hat{\mathbf{q}} \\
		&\dot{\hat{\mathbf{q}}}=-\nabla \mathbf{g}(\hat{\mathbf{p}})-\mathbf{G}(\hat{\mathbf{p}})^T\bm{\lambda}-\gamma\hat{\mathbf{q}}\\
		&\bm{\phi}(\hat{\mathbf{p}})=\bm{\rho}
	\end{align}
\end{subequations}
where $\gamma > 0$ is the dissipative coefficient, and  \(\hat{\mathbf{p}}\) is the precoder vector. 
Equation \eqref{30} is a dissipative version of (\ref{LDE}). 
The dynamics resulting from \eqref{30} converges to a stationary point of \eqref{WSR}.

Let $\Phi_t^C(\hat{\mathbf{p}},\hat{\mathbf{q}}): \mathcal{M} \rightarrow \mathcal{M}$ and $\Phi_t(\hat{\mathbf{p}},\hat{\mathbf{q}}): \mathcal{M} \rightarrow \mathcal{M}$ denote the flow of dynamical systems in \eqref{LDE} and \eqref{30}, respectively.
The flow $\Phi_t$ can be found by composing  $\Phi_t^C$ and $\Phi_t^D$, where
$\Phi_t^D(\hat{\mathbf{p}},\hat{\mathbf{q}})=(\hat{\mathbf{p}},e^{-\gamma t}\hat{\mathbf{q}})$. To obatin a practical algorithm, we approximate the flow $\Phi_t$ with a numerical method in the following. 

\subsection{Discrete Dissipative Hamiltonian Dynamics for Precoder Design}

The dissipative augmented Hamiltonian system in the previous section is a continuous one. To obtain a numerical method, the RATTLE integrator is utilized to discretize the continuous system to obtain an approximate flow of \eqref{30} \cite{ghirardelli2023optimization}. RATTLE is a step-by-step iterative method for numerically solving constrained dynamics problems at a time step.
A discrete approximation $\Phi_h^C$ of $\Phi_t^C$ with the RATTLE is defined as
\begin{subequations}
	\begin{align}
		&\hat{\mathbf{q}}_{n+1/2}= \hat{\mathbf{q}}_{n}-\frac{h}{2}(\nabla \mathbf{g}(\hat{\mathbf{p}}_n)+\bm{\mathbf{G}(\hat{\mathbf{q}}_n)^T\lambda}_{n})\label{diedai11} \\
		&\hat{\mathbf{p}}_{n+1}=\hat{\mathbf{p}}_{n}+h\mathbf{M}^{-1}\hat{\mathbf{q}}_{n+\frac12} \label{diedai12} \\
		&\hat{\mathbf{q}}_{n+1}= (\hat{\mathbf{q}}_{n+1/2}-\frac{h}{2}(\nabla \mathbf{g}(\hat{\mathbf{p}}_{n+1})+\mathbf{G}(\hat{\mathbf{p}}_{n+1})^T\bm{\mu}_n))  \label{diedai13}
	\end{align}
\end{subequations}
where $h$ denotes the step length. 

For the dissipative Hamiltonian system, an approximation $\Phi_h$ of $\Phi_t$ is obtained as a symmetric leapfrog
composition as $\Phi_h=\Phi_{h/2}^D\circ\Phi_h^C\circ\Phi_{h/2}^D$. 
Then, the position and momentum coordinates ($\hat{\mathbf{p}}$,  $\hat{\mathbf{q}}$) are iteratively obtained as \cite{ghirardelli2023optimization}
\begin{subequations}\label{diedaizong}
	\begin{align}
		&\hat{\mathbf{q}}_{n+1/2}=e^{-\gamma h/2}\hat{\mathbf{q}}_{n}-\frac{h}{2}(\nabla \mathbf{g}(\hat{\mathbf{p}}_n)+\bm{\mathbf{G}(\hat{\mathbf{q}}_n)^T\lambda}_{n}) \label{diedai21} \\
		&\hat{\mathbf{p}}_{n+1}=\hat{\mathbf{p}}_{n}+h\mathbf{M}^{-1}\hat{\mathbf{q}}_{n+\frac12} \label{diedai22} \\
		&\hat{\mathbf{q}}_{n+1}=e^{-\gamma h/2}(\hat{\mathbf{q}}_{n+1/2}-\frac{h}{2}(\nabla \mathbf{g}(\hat{\mathbf{p}}_{n+1})+\mathbf{G}(\hat{\mathbf{p}}_{n+1})^T\bm{\mu}_n)). \label{diedai23}
	\end{align}
\end{subequations}
Now, the only thing unknown is the calculation of \(\bm{\mu}_n\).  For simplicity, we set $\mathbf{M}=\mathbf{I}$ hereafter. 
The constraint becomes $\mathbf{G}(\hat{\mathbf{p}}_{n+1})\hat{\mathbf{q}}_{n+1}=0$, From the constraint, we derive the calculation of \(\bm{\mu}_n\) in the following theorem.
\begin{theorem}\label{thm2}
	Expression of $\bm{\mu}_n$ satisfies \eqref{26a} and \eqref{diedai23} is given by
	\begin{align}
		\bm{\mu}_n=\mathbf{C}^{-1}\frac{2\mathbf{G}(\hat{\mathbf{p}}_{n+1})\hat{\mathbf{q}}_{n+1/2}-h\mathbf{G}(\hat{\mathbf{p}}_{n+1})\nabla\mathbf{g}(\hat{\mathbf{p}}_{n+1})}{h} 
	\end{align}  
	where $\mathbf{C}=\mathrm{diag}(\bm{\rho})$
	denotes precoder vector power factor.
\end{theorem} 
\begin{proof}
	See in Appendix \ref{fulu2}.
\end{proof}
Choosing step sizes is a critical task of the optimization method, as it directly affects both the convergence and the complexity of the algorithm per iteration \cite{nocedal1999numerical,boyd2004convex}. 
Line search methods are commonly used methods to determine the step size with fixed search direction. 
However, the main challenge is to balance the sufficient decrease in the objective function with the additional computational cost introduced by the step size search\cite{nocedal1999numerical}.
Both the Armijo rule\cite{armijo1966minimization} and Wolfe\cite{wolfe1969convergence} conditions ensure convergence in line search procedures.
However, since multiple evaluations of the objective function are required at each iteration to satisfy the conditions, the computational complexity becomes significantly high, especially in the considered UCN massive MIMO system.

We employ a proportional controller  \cite{wadia2021optimization} to adaptively adjust the step size based on
a parameter $\delta_n$.
The step size update rule is given by  \cite{wadia2021optimization}
\begin{align}
	h_{n+1}=(\frac r {\delta_{n+1}})^{\frac{\theta}{2}}h_n \label{step}
\end{align}
where $r$ and $\theta$ denote hyper-parameters.
When $\theta = 0$, the method reduces to that using a fixed step size.
The \(\delta_n\) represents the current error and serves as an indicator of the accuracy or reliability of the current gradient update step, which can be obtained as
\begin{align}
\delta_{n+1} = \left\| \hat{\mathbf{p}}_{n} - \hat{\mathbf{p}}_{n+1} - \frac{1}{2}h \left( \nabla \mathbf{g}(\hat{\mathbf{p}}_{n}) + \nabla \mathbf{g}(\hat{\mathbf{p}}_{n+1}) \right) \right\|.
\end{align}

\begin{figure}[tbp]
	\begin{algorithm}[H]
		\caption{Precoder Design for UCN Massive MIMO System with Symplectic Optimization}\label{alg}
		\begin{algorithmic}[1]
			\REQUIRE Initialize step length h, the precoder $\hat{\mathbf{p}}_{0}$ and the momentum coordinate $\hat{\mathbf{q}}_{0}$
			\ENSURE precoder $\hat{\mathbf{p}}$
			\STATE \textbf{repeat}
			\STATE Calculate Lagrange multiplier  $\bm{\lambda}_n$ as
			\[{\bm{\lambda}_n}=\mathbf{C}^{-1}\Big(\mathbf{Q}_n\hat{\mathbf{q}}_n -\mathbf{G}(\hat{\mathbf{p}}_n)\nabla\mathbf{g}(\hat{\mathbf{p}}_n)\Big)\]
			\STATE Calculate the momentum $\hat{\mathbf{q}}_{n+1/2}$ as 
			\[\hat{\mathbf{q}}_{n+1/2}=e^{-\gamma h/2}\hat{\mathbf{q}}_{n}-\frac{h}{2}(\nabla \mathbf{g}(\hat{\mathbf{p}}_n)+\bm{\mathbf{G}(\hat{\mathbf{q}}_n)^T\lambda}_{n}).\]
			\STATE Compute the precoder $\hat{\mathbf{p}}_{n+1}$ as
			\[\hat{\mathbf{p}}_{n+1}=\hat{\mathbf{p}}_{n}+h \hat{\mathbf{q}}_{n+\frac12}.\]
			\STATE Update $\nabla g(\hat{\mathbf{p}}_{n})$ according to (\ref{Hp}).
			\STATE Calculate $\bm{\mu}_n$ as
			\[	\bm{\mu}_n=\mathbf{C}^{-1}\frac{2\mathbf{G}(\hat{\mathbf{p}}_{n+1})\hat{\mathbf{q}}_{n+1/2}-h\mathbf{G}(\hat{\mathbf{p}}_{n+1})\nabla\mathbf{g}(\hat{\mathbf{p}}_{n+1})}{h} \]
			\STATE Update the momentum $\hat{\mathbf{q}}_{n+1}$ as
			\[\hat{\mathbf{q}}_{n+1}=e^{-\gamma h/2}(\hat{\mathbf{q}}_{n+1/2}-\frac{h}{2}(\nabla \mathbf{g}(\hat{\mathbf{p}}_{n+1})+\mathbf{G}(\hat{\mathbf{p}}_{n+1})^T\bm{\mu}_n)).\]
			\STATE Compute the step size as
			\[	h_{n+1}=(\frac r {\delta_{n+1}})^{\frac{\theta}{2}}h_n.\]
			\STATE Update the iteration as $d=d+1$.
			\STATE \textbf{until} convergence
		\end{algorithmic}
	\end{algorithm}
\end{figure}

A larger \(\theta\) increases the aggressiveness of the step size adaptation, which may accelerate convergence by reducing the required number of iterations\cite{wadia2021optimization}.
The hyper-parameter \(r\) controls the discrepancy between the two update schemes.
If the current error \(\delta_{n}\)
exceeds \(r\), the step size is decreased; otherwise, it is increased. This forms an adaptive adjustment strategy that ensures each update is both stable and as large as possible, thereby improving optimization efficiency.

Algorithm 1 summarizes the proposed symplectic optimization for the precoder design in the considered UCN massive MIMO system.

\subsection{Complexity Analysis}

The computational complexity of Algorithm \ref{alg} per iteration is primarily dominated by the calculation of $\nabla\mathbf{g}(\hat{\mathbf{p}})$ in \eqref{Hp}. 
The user-centric rule reduces the system dimension, and the corresponding complexity analysis is presented as follows.

For the interference term \(r_k\) in \eqref{rk}, we consider all UTs \(t \ne k\), where each contributes a term of
\(\sum_{m \in \mathcal{B}_t} \hat{\mathbf{p}}_{m,t}^T \check{\mathbf{H}}_{m,k}\), whose computational complexity is \(\mathcal{O}(|\mathcal{B}_t|M_t)\). 
Since the summation involves \(|\mathcal{B}_t|\) BSs, the overall complexity for computing \(r_k\) is 
\(\mathcal{O}\left( M_t \sum_{ t \ne k} |\mathcal{B}_t| \right)\), where \(|\mathcal{B}_t|\) denotes the number of serving cluster of UT \(t\).

To evaluate the gradient \(\nabla\mathbf{g}(\hat{\mathbf{p}})\), \(r_k\) needs to be computed for all UTs \(k = 1, \dots, K\). 
Therefore, the total computational complexity for computing all \(r_k\) is
\begin{align}
	\mathcal{O}\left( M_t  (K - 1) \sum_{t = 1}^{K} |\mathcal{B}_t| \right).
\end{align}

For the first term of \eqref{Hp}
\[
- r_k^{-1} b_k \check{\mathbf{H}}_{l,k} \left( \sum_{l \in \mathcal{B}_k} \check{\mathbf{H}}_{l,k}^T \hat{\mathbf{p}}_{l,k} \right)
\]
each inner product \(\check{\mathbf{H}}_{l,k}^T \hat{\mathbf{p}}_{l,k}\) requires \(\mathcal{O}(M_t)\) operations,  and the summation \(\sum_{l \in \mathcal{B}_k} \check{\mathbf{H}}_{l,k}^T \hat{\mathbf{p}}_{l,k}\) needs to be computed once for each BS in the cluster \(\mathcal{B}_k\).
So the computational complexity of this term is 
\[\mathcal{O}(|\mathcal{B}_k|M_t)\]
where \(|\mathcal{B}_k|\) denotes the number of serving cluster of UT \(k\).

Note that the scalar term \(a_k\) in \eqref{ak} shares the same intermediate result as the first term in \eqref{Hp}, and thus does not incur additional computational cost when reused.
The term \(b_k\) in \eqref{bk} only involves scalar operations based on the previously computed values of \(a_k\) and \(r_k\), and hence has negligible complexity.
  
For the second term of \eqref{Hp}
\[
\sum_{t\neq k}^Kr_t^{-2}a_{t}b_{t}\check{\mathbf{H}}_{l,t}(\sum_{l\in \mathcal{B}_k}\check{\mathbf{H}}_{l,t}^T\hat{\mathbf{p}}_{l,k})
\]
which involves a summation over \(t \neq k\), for each \(t\), we need to calculate \(\check{\mathbf{H}}_{l,t}^T \hat{\mathbf{p}}_{l,k}\), and its complexity is \(\mathcal{O}(M_t)\). 
Like the first term, the computational complexity of \(\sum_{l \in \mathcal{B}_k} \check{\mathbf{H}}_{l,t}^T \hat{\mathbf{p}}_{l,k}\) is \(\mathcal{O}(|\mathcal{B}_k| M_t)\). 
Since the summation is from $1$ to \(K\) for \(t\) (excluding \(k\)), the complexity of computing the second term is \[\mathcal{O}((K-1) |\mathcal{B}_k|M_t).\]

For full gradient  $\nabla\mathbf{g}({\hat{\mathbf{p}}})=[\nabla\mathbf{g}(\hat{\mathbf{p}}_{1})^T,\cdots,\nabla\mathbf{g}(\hat{\mathbf{p}}_{L})^T]^T$, the computation is repeated multiple times across all UTs and BSs, so the total complexity is 
\begin{align}
	\mathcal{O}\big(N_S (\sum_{l \in \mathcal{S}_B} \sum_{k \in \mathcal{U}_l} K  |\mathcal{B}_k|  M_t +  M_t  (K - 1) \sum_{t = 1}^{K} |\mathcal{B}_t| )\big)
\end{align}
where $N_S$ denotes the iteration number of the symplectic method.

In \cite{wu2018precoder}, the WMMSE precoder for the coordinated multi-point (CoMP) joint transmission (JT) scenario is derived, making it compatible with UCN massive MIMO systems.  
The computational complexity of the WMMSE precoder in UCN massive MIMO system is \(	\mathcal{O}\Big(N_W\big((4K + 1) M_t \sum_{l \in \mathcal{S}_B} U_l+BM_t^3+
B\sum_{l \in \mathcal{S}_B}U_lM_t^2+\sum_{l \in \mathcal{S}_B}U_lM_t^2\big)\Big)\) \cite{sun2025precoder}, where $N_W$ denotes iteration numbers of WMMSE method.

Obviously, as the system dimensions and the number of BSs serving UTs increase, the symplectic method has much lower complexity compared to 
 the WMMSE method when their numbers of iterations are the same.

\section{Simulation Results}
\label{section4}
In this section, we evaluate the proposed precoder design for UCN massive MIMO systems with symplectic optimization. The widely used QuaDRiGa channel model \cite{jaeckel2014quadriga} is adopted to generate the channels. The “3GPP\_38.901\_UMa\_NLoS”\cite{jaeckel2014quadriga}  scenario is considered. 
Moreover, a tri-sector configuration is adopted for each gNodeB(gNB) to ensure better coverage \cite{3gpp.38.300}, and the number of gNBs deployed in the system is $7$. 
As illustrated in Fig. \ref{BS}, each gNB consists of $3$ BSs, and each BS covers a $120^\circ$ sector \cite{3gpp.38.104}.
Thus, the system consists of a total of \(B = 21\) BSs. The number of antennas for the UPA equipped in each BS is $N_t = 8 \times 16$. The number of single antenna UTs is \(U = 300\), and they are randomly generated within a circle of radius $1000m$.
For each UT,  a cluster of BSs that offer the best channel quality are selected.
The heights of each BS and UT are $25m$ and $1.5m$, respectively. 
For simplicity, we assume $w_1 = w_2 = \cdots = w_{S_{U}} = 1$.
To reproduce the realistic scenario, the NLOS model is used to account for large-scale decay. 
The number of BSs in the clusters for each UT are assumed to be the same, \textit{i.e.}, $B_1 = B_2 = \cdots = B_U = B_{\mathrm{sc}}$. The center frequency is $6.7$ GHz.
Main simulation parameters are also provided in Table \ref{table1}.
\begin{figure}[btp]
	\centerline{\includegraphics[width=1\columnwidth]{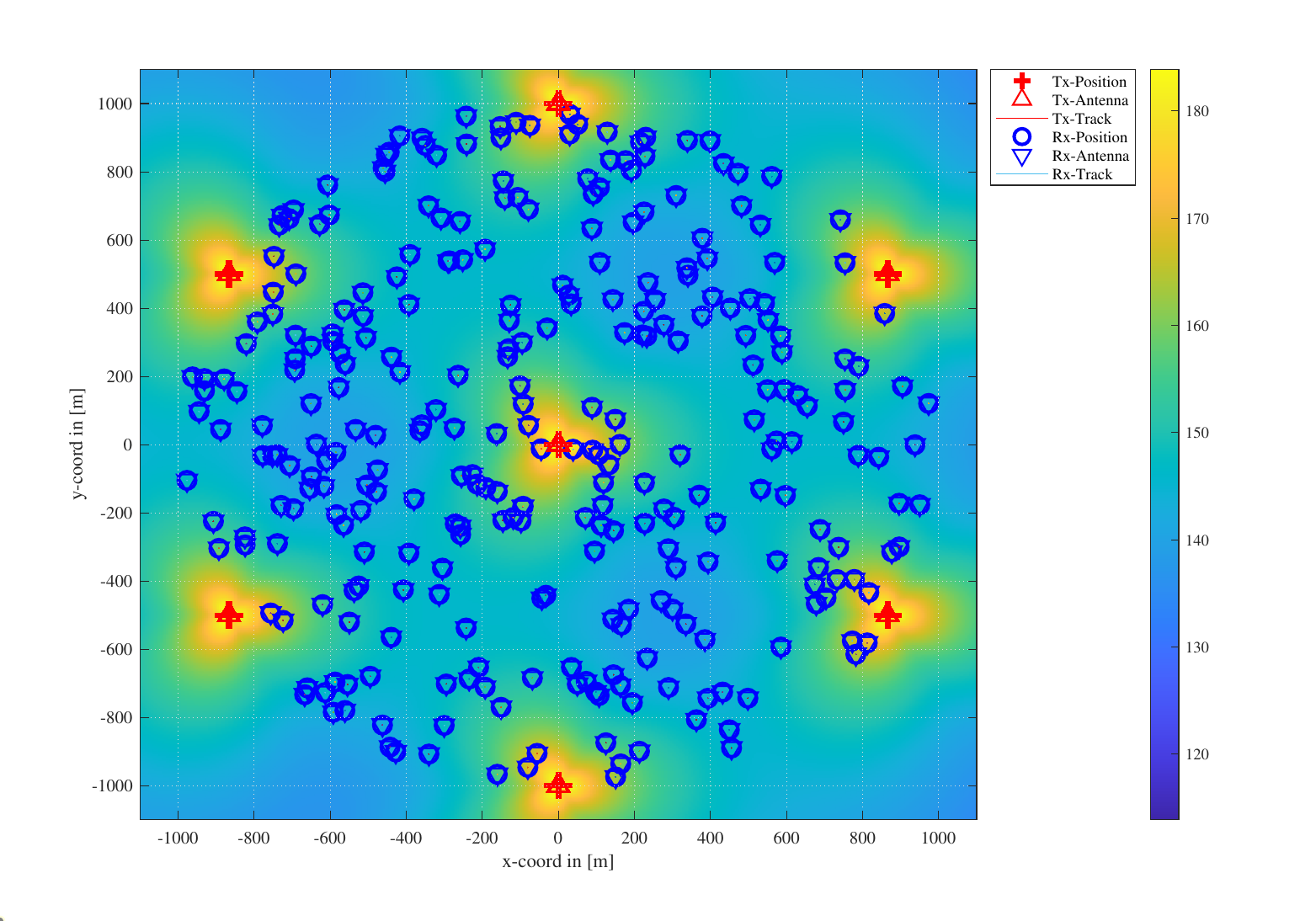}}
	\caption{The layout of the UCN massive MIMO system.}
	\label{BS}
\end{figure}

In the UCN massive MIMO system, each UT needs to be served by a different set of BSs, which necessitates the design of a BS scheduling strategy. 
In this paper, we adopt a reference signal received power (RSRP)-based UT association scheme to assign each UT to a primary serving BS.
The RSRP between each BS and each UT is calculated as follow
\begin{align}
\mathrm{RSRP}_{l,k} = 10 \operatorname{log}(\mathbf{h}_{l,k}^H\mathbf{h}_{l,k}).
\end{align}
Each UT is assigned to the BS with the highest RSRP as its primary serving BS.
Subsequently, the \(\Delta \mathrm{RSRP}_{l,k}\) is computed to assist in selecting additional serving BSs for each UT.
It can be obtained as
\begin{align}
	\Delta\mathrm{RSRP}_{l,k} & =\max_{l \in \mathcal{B}_k}\mathrm{RSRP}_{l,k}-\mathrm{RSRP}_{l,k}
\end{align}
 which denotes the power gap between a candidate BS and the UT's primary serving BS.
 Finally, by sorting the \(\Delta\mathrm{RSRP}_{l,k}\) values in ascending order, a subset of BSs are selected to serve each UT.
\begin{table}[tbp]
	\caption{Parameter Settings}
	\label{table1}
	\begin{center}
		\begin{tabular}{cc}
			\toprule  
			Parameter &Value\\
			\hline
			Scenario &3GPP\_38.901\_UMa\_NLoS \\
			Center frequency &6.7GHz\\
			Number of BS antennas $M_v \times M_h$ &8 $\times$ 16\\
			Number of UT antennas $M_r$ &1 \\
			Number of subcarriers $M_c$ &2048\\
			Subcarrier spacing $\Delta f$&30KHz\\
			Speed of each UT &3 km/h\\
			$\sigma_z^2$ &-104 dBm\\
			\bottomrule
		\end{tabular}
	\end{center}
\end{table}

The performance comparison between the proposed precoder, the RZF precoder and the WMMSE precoder on the sum-rate metrics with \(B_{sc} = 21\) is shown in Fig. \ref{performance}. 
Both the symplectic optimization and the WMMSE precoder methods are executed with \(N = 50\) iterations.
We use two kinds of initial values to better verify the performance of the proposed method. 
Notably, the symplectic optimization involves no matrix inversion with random initial value.
Simulation results show that the proposed method outperforms the WMMSE precoder in the whole transmit power regimes.
The proposed method has a 29.3\% performance gain with RZF initial value and 13.2\% performance gain with random initial value when $P = 24 \  \mathrm{dBm}$. 
The WMMSE precoder needs the initial values of RZF, which increase the computational complexity.
It also illustrates that the proposed method is much better than the RZF method, which requires the least computational complexity. 
Using random initial values avoids the matrix inversion required by the RZF precoders, and the performance becomes worse but still better than the WMMSE precoder. 
Since the proposed method has much lower complexity per iteration compared to the WMMSE precoder, the comparison shows the efficiency of the symplectic method is high.

\begin{figure}[htbp]
	\centerline{\includegraphics[width=1\columnwidth]{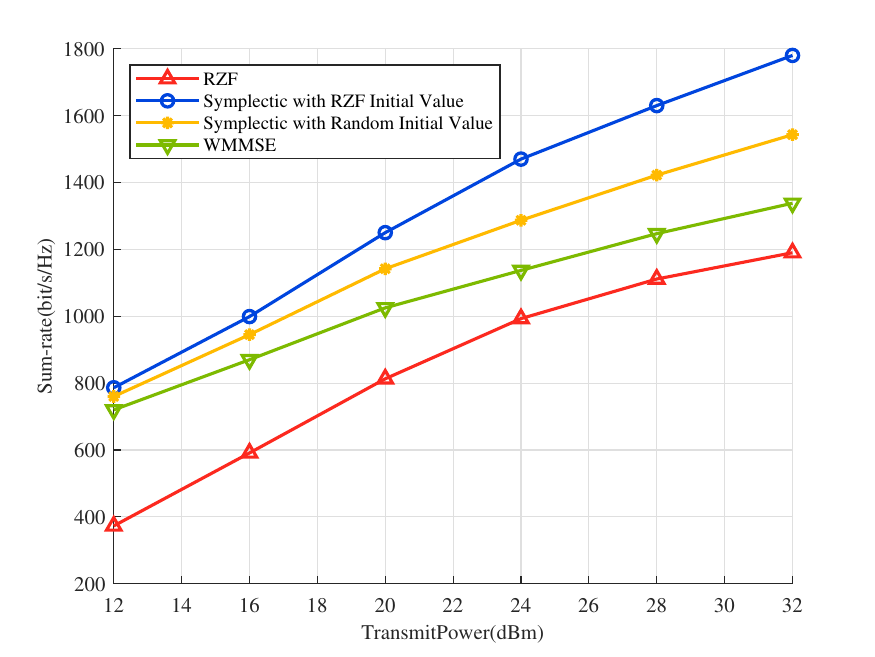}}
	\caption{Comparison of the sum-rate performance of the symplectic method with other methods.}
	\label{performance}
\end{figure}

Next, we investigate the convergence behavior of the proposed symplectic method in different transmit power regimes for precoder design in the UCN massive MIMO system. 
In Fig. \ref{convergence}, we plot the sum-rate performance over iteration number of the symplectic method when $P = 16, 20, 24 \ \mathrm{dBm}$, and all cases employ a well-matched adaptive step length and RZF initial value. 
Observing Fig. \ref{convergence}, we have that all three power levels achieve convergence within \(N = 42\) iterations. Moreover, the proposed method converge rapidly at low transmit power, and 18 iterations are enough for the proposed method to converge at $P = 16\  \mathrm{dBm}$.
As transmit power increases, more iterations are required.
Specifically, \(42\) iterations are required to ensure convergence at $P = 24 \ \mathrm{dBm}$.

\begin{figure}[tbp]
	\centerline{\includegraphics[width=1\columnwidth]{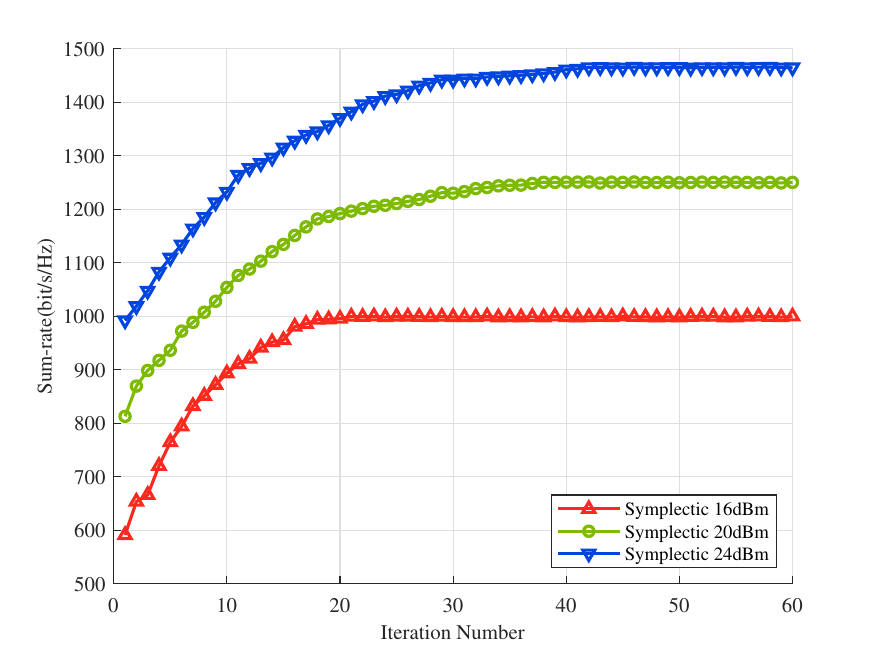}}
	\caption{Convergence behavior of the symplectic method in different power schemes.}
	\label{convergence}
\end{figure}

We then investigate the convergence behavior of the proposed symplectic method with different step sizes in the UCN massive MIMO. The results are plotted in Fig. \ref{differenth}. It can be observed that the proposed method converges in \(N = 37\) iterations with adaptive step size, whereas \(N = 52\) iterations are required when using a fixed step size of $h=0.01$, and \(N = 67\) iterations are needed for $h=0.004$.
We can also see that using an adaptive step size improves the sum-rate performance by 3.57\% compared to a fixed step size of $h=0.04$.
The simulation results demonstrate that the number of iterations required by the proposed method is reduced and the sum-rate performance can also be improved by using adaptive step size.

\begin{figure}[tbp]
	\centerline{\includegraphics[width=1\columnwidth]{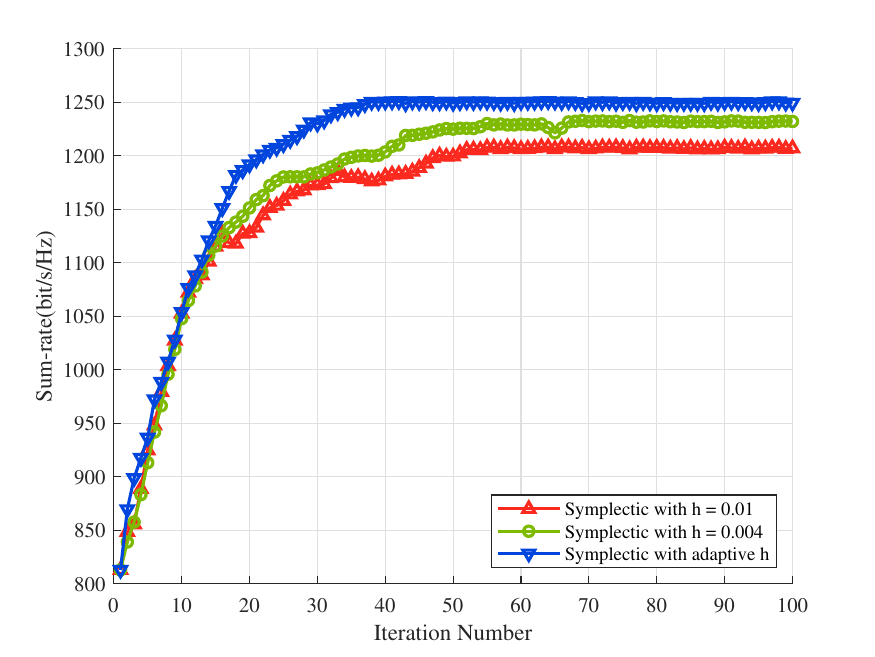}}
	\caption{Convergence behavior of the symplectic method with different step sizes.}
	\label{differenth}
\end{figure}

Fig. \ref{SwithGradient} plot the convergence comparison of the proposed symplectic optimization method, the GD method, and the NAGD method in the UCN massive MIMO system.
All three methods are initialized using the RZF precoder.
Both the NAGD and conventional GD methods determine the step size with the line search to ensure sufficient decrease in the objective function, but this comes at the cost of significantly increased computational complexity.
It can be observed that at the transmit power of $P = 16\  \mathrm{dBm}$ and $P = 24\  \mathrm{dBm}$, the proposed method converges faster than the other two methods and achieves better sum-rate performance.
At a transmit power of $P = 24\  \mathrm{dBm}$, the proposed method achieves a sum-rate performance gain of 6.86\% over NAGD method and 20.12\% over conventional GD method at the $N = 20$ iteration.
We also have that the proposed method achieves a 3.97\% improvement in final WSR performance compared to NAGD, while the conventional GD method fails to converge even after 60 iterations at $P = 24\  \mathrm{dBm}$.
Due to the use of line search in NAGD and GD methods, which leads to high per-iteration computational complexity, the simulation results indicate that the proposed method achieves faster convergence while maintaining lower complexity.

\begin{figure}[tbp]
	\centerline{\includegraphics[width=1\columnwidth]{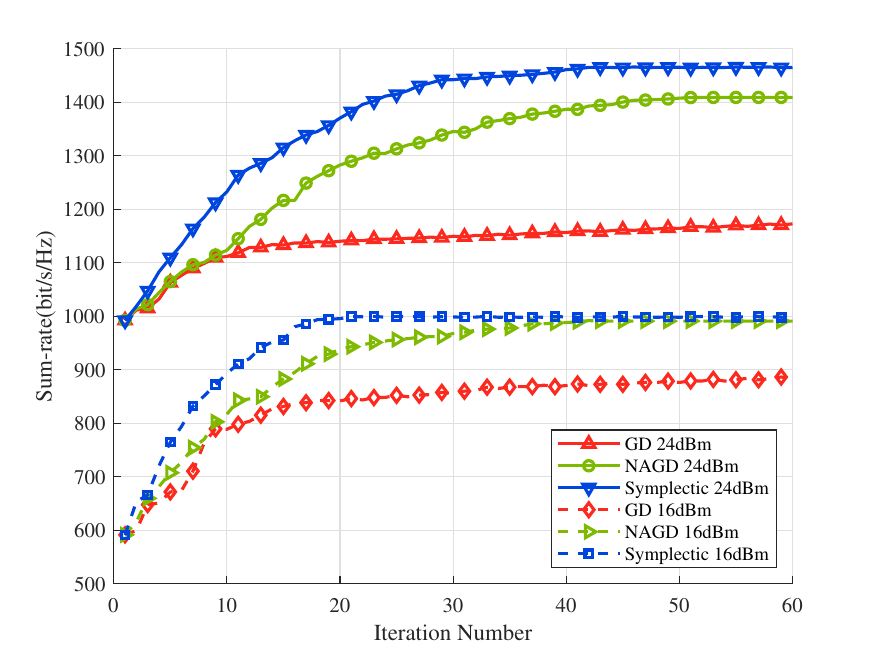}}
	\caption{Convergence comparison of the proposed symplectic optimization method, the GD method, and the NAGD method with RZF initials and two power schemes.}
	\label{SwithGradient}
\end{figure}

To further analyze the convergence performance, we compare the proposed method with the conventional GD and NAGD methods for precoder design in the considered system in Fig.~\ref{SwithGradient_0Ini}.
Both the conventional GD and NAGD methods employ line search with the Armijo condition to adjust the step size, while the proposed method adopts an adaptive step size strategy with \eqref{step} in the simulations.  
Moreover, all methods are initialized with the same randomly generated initial values to ensure a fair comparison. From Fig.~\ref{SwithGradient_0Ini}, we have that the proposed method converges faster and outperforms both GD and NAGD methods at the transmit power of $P = 16\  \mathrm{dBm}$ and $P = 24\  \mathrm{dBm}$.
The proposed method with random initialization converges in $N = 53$ iterations at $P = 24\  \mathrm{dBm}$, which is slightly slower than with RZF initialization that converges in $N = 40$ iterations.
Compared to RZF initialization, random initialization results in a 15.1\% performance degradation at $P = 24\  \mathrm{dBm}$, but completely avoids matrix inversion, thereby further reducing computational complexity as the numbers of UTs and BSs increase.
At $P = 24\  \mathrm{dBm}$, the proposed method achieves a 9.01\% performance gain over the NAGD method at convergence, and yields a 44.41\% improvement compared to the conventional GD method after $N = 100$ iterations.
In conclusion, simulation results show that the proposed precoder design converges faster and achieves better sum-rate performance compared to the other two methods.

\begin{figure}[tbp]
	\centerline{\includegraphics[width=1\columnwidth]{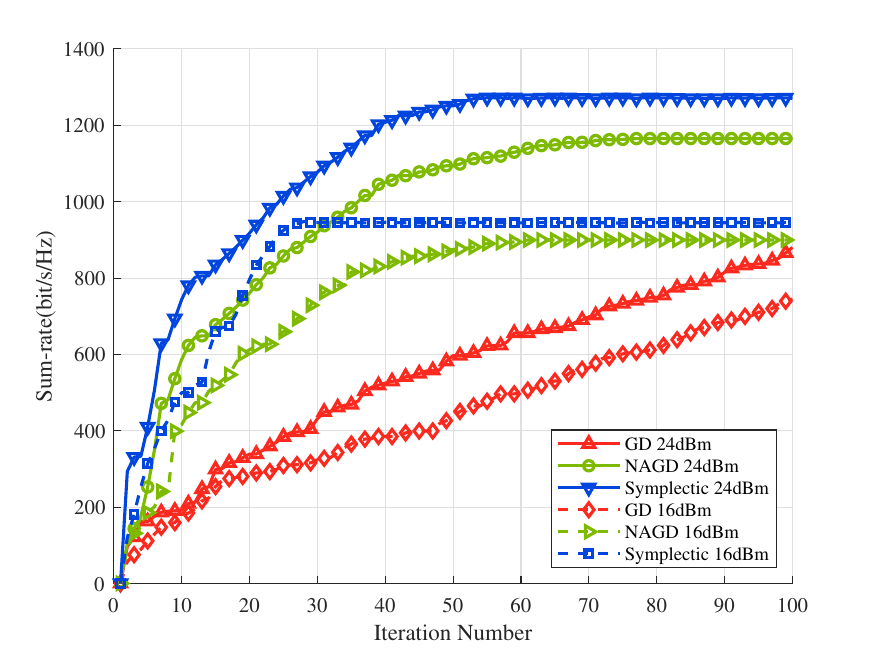}}
	\caption{Convergence comparison of the proposed symplectic optimization method, the GD method, and the NAGD method with random initials and two power schemes.}
	\label{SwithGradient_0Ini}
\end{figure}

Finally, we compare the sum-rate performance of the proposed symplectic optimization with the WMMSE precoder under different serving cluster sizes. 
The UCN massive MIMO system with $B_{\mathrm{sc}} = 1$ and $B_{\mathrm{sc}} = 21$ are viewed as a cellular and a conventional network massive MIMO system, respectively.
In Fig. \ref{cluster}, the proposed symplectic optimization method consistently outperforms the WMMSE method across all three cluster configurations. 
The symplectic optimization method with a large serving cluster size ($B_{\mathrm{sc}} = 21$) has the highest overall performance.
However, its computational complexity is much higher than that of the method with $B_{\mathrm{sc}} = 3$. 
The proposed method with \( B_{\mathrm{sc}} = 3 \) achieves an 86.6\% performance gain compared to \( B_{\mathrm{sc}} = 1 \), while increasing \( B_{\mathrm{sc}} \) further to 21 yields only an additional 28.9\% improvement. 
This indicates that most of the WSR performance benefits can be achieved with a relatively small serving cluster size, with significantly reduced computational complexity.
In addition, the symplectic optimization method with $B_{\mathrm{sc}} = 3$ surpasses the WMMSE method with $B_{\mathrm{sc}} = 21$ in the high transmit power regime, highlighting its effectiveness and computational efficiency.

\begin{figure}[tbp]
	\centerline{\includegraphics[width=1\columnwidth]{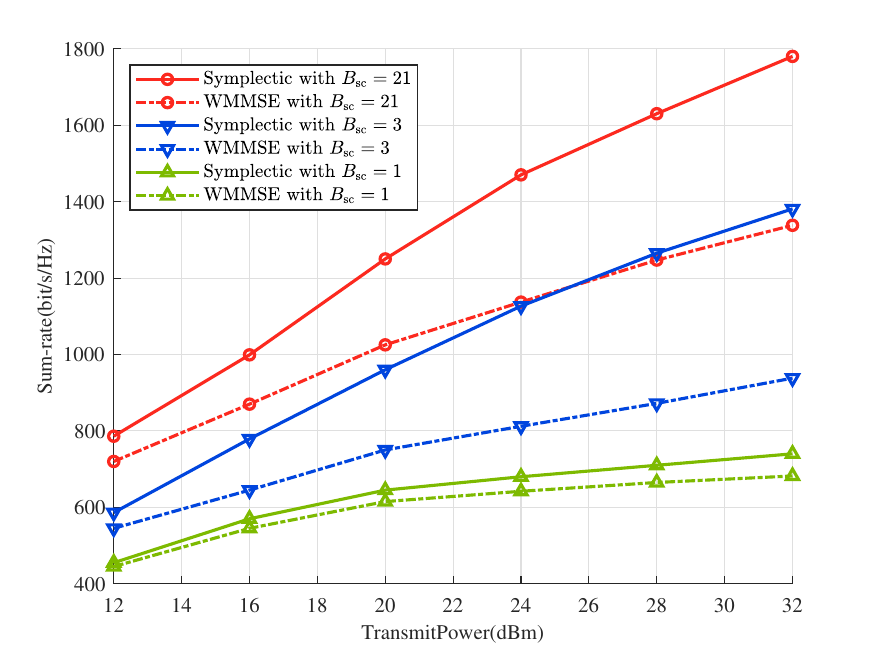}}
	\caption{The sum-rate performance comparison of the symplectic method and the WMMSE method with different sizes of serving cluster.}
	\label{cluster}
\end{figure}

\section{Conclusion}
\label{section5}
In this paper, we propose a precoder design for UCN massive MIMO system by using a symplectic optimization approach. 
In the UCN massive MIMO system, the dimensions of the precoders are reduced compared to conventional network massive MIMO system.
To further alleviate the computational burden caused by the matrix inversion, we employ symplectic optimization, which is related to a dissipative dynamical system. 
Then, we transform the received model into the real field, and derive an iterative symplectic method for the precoder design.  
The precoder is directly obtained through he iterative process.
Moreover, the complexity analysis of the proposed method is provided to demonstrates its high computational efficiency. 
The simulation results show that the proposed precoder design converges faster and achieves significant performance gains with lower complexity compared to the WMMSE precoder.


\appendices
\section{Proof for Theorem \ref{thm1}}
\label{fulu1}
Let $\hat{\mathbf{q}}_{l}=[\hat{\mathbf{q}}_{l,1}^T,\hat{\mathbf{q}}_{l,2}^T,\dots,\hat{\mathbf{q}}_{l,K}^T]^T\in \mathbb{R}^{2MK\times1}$ for simplicity, and thus the momentum $\hat{\mathbf{q}}$ can be written as $\hat{\mathbf{q}}=[\hat{\mathbf{q}}_1^T,\hat{\mathbf{q}}_2^T,\dots,\hat{\mathbf{q}}_L^T]^T\in \mathbb{R}^{2MKL\times 1}$.
The formulation of \eqref{26b} can be computed as
\begin{align}
\mathbf{G}(\hat{\mathbf{p}})\hat{\mathbf{q}}&= 
\begin{bmatrix}
  \hat{\mathbf{p}}_1^T &  &  & \\
  &  \hat{\mathbf{p}}_2^T&  & \\
  &  & \ddots  & \\
  &  &  &\hat{\mathbf{p}}_L^T
\end{bmatrix}\hat{\mathbf{q}}.
\end{align}
By substituting the forms of $\hat{\mathbf{p}}_l^T$ and $\hat{\mathbf{q}}$, we then have 
\begin{align}
\mathbf{G}(\hat{\mathbf{p}})\hat{\mathbf{q}} 
&=\begin{bmatrix}
\overbrace{\hat{\mathbf{p}}_{1,1}^T ,\dots, \hat{\mathbf{p}}_{1,K}^T}^{\hat{\mathbf{p}}_{1}^T}&  & \\
 & \ddots  & \\
 &  &\overbrace{\hat{\mathbf{p}}_{L,1}^T ,\dots, \hat{\mathbf{p}}_{L,K}^T}^{\hat{\mathbf{p}}_{L}^T}
\end{bmatrix}\begin{bmatrix}
    \hat{\mathbf{q}}_{1,1}\\
    \vdots\\
    \underbrace{\hat{\mathbf{q}}_{1,K}}_{\hat{\mathbf{q}}_1}\\
    \vdots\\
    \hat{\mathbf{q}}_{L,1}\\
    \vdots\\
    \underbrace{\hat{\mathbf{q}}_{L,K}}_{\hat{\mathbf{q}}_L}\\
\end{bmatrix}\nonumber\\
&=\begin{bmatrix}
  \hat{\mathbf{p}}_1^T&  &  & \\
  &  \hat{\mathbf{p}}_2^T&  & \\
  &  & \ddots  & \\
  &  &  &\hat{\mathbf{p}}_L^T
\end{bmatrix}\begin{bmatrix}
\hat{\mathbf{q}}_1 \\
 \hat{\mathbf{q}}_2\\
\vdots \\
\hat{\mathbf{q}}_L
\end{bmatrix}\nonumber\\
&=\begin{bmatrix}
\hat{\mathbf{p}}_1^T\hat{\mathbf{q}}_1 \\
 \hat{\mathbf{p}}_2^T\hat{\mathbf{q}}_2\\
\vdots \\
\hat{\mathbf{p}}_L^T\hat{\mathbf{q}}_L
\end{bmatrix}\in\mathbb{R}^{L\times 1}\label{Gq}.
\end{align}
Differentiating \eqref{Gq} with respect to $\hat{\mathbf{p}}$, we can obtain that
\begin{align}
    \frac{\partial}{\partial\hat{\mathbf{p}}}(\mathbf{G}(\hat{\mathbf{p}})\hat{\mathbf{q}})&=\begin{bmatrix}
    	\frac{\partial\hat{\mathbf{p}}_1^T\hat{\mathbf{q}}_1}{\partial\hat{\mathbf{p}}}\ \\
    	\frac{\partial\hat{\mathbf{p}}_2^T\hat{\mathbf{q}}_2}{\partial\hat{\mathbf{p}}}\\
    	\vdots \\
    	\frac{\partial\hat{\mathbf{p}}_L^T\hat{\mathbf{q}}_L}{\partial\hat{\mathbf{p}}}
    \end{bmatrix}\nonumber\\
    &=\begin{bmatrix}
	 \hat{\mathbf{q}}_1^T  &  &  & \\
	  & \hat{\mathbf{q}}_2^T &  & \\
	  &  & \ddots  & \\
	  &  &  &\hat{\mathbf{q}}_L^T
	\end{bmatrix}\in\mathbb{R}^{L\times  2MKL} \nonumber\\
	&= \mathrm{vecdiag}(\hat{\mathbf{q}}_1^T, \dots, \hat{\mathbf{q}}_L^T) \nonumber \\
	&\triangleq \mathbf{Q}.
\end{align}
Thus, equation (\ref{26b}) can be rewritten as
\begin{align}
	\mathbf{M}^{-1}\mathbf{Q}\dot{\hat{\mathbf{p}}}+\mathbf{G}(\hat{\mathbf{p}})\mathbf{M}^{-1}\dot{\hat{\mathbf{q}}}=0. \label{52}
\end{align}
In equation \eqref{LDE2}, we have
\begin{align}
	&\dot{\hat{\mathbf{q}}}=-H_{\mathbf{p}}-\mathbf{G}(\hat{\mathbf{p}})^T\bm{\lambda}  = -\nabla \mathbf{g}(\hat{\mathbf{p}})-\mathbf{G}(\hat{\mathbf{p}})^T\bm{\lambda}.
\end{align}
Therefore, substituting the expression of \(\dot{\hat{\mathbf{q}}}\) into equation \eqref{52}, we obtain
\begin{align}
&\mathbf{M}^{-1}\mathbf{Q}\dot{\hat{\mathbf{p}}}-\mathbf{G}(\hat{\mathbf{p}})\mathbf{M}^{-1}\Big(\nabla\mathbf{g}(\hat{\mathbf{p}})+\mathbf{G}(\hat{\mathbf{p}})^T\bm{\lambda}\Big) \nonumber \\ 
=&\mathbf{M}^{-2}\mathbf{Q}\hat{\mathbf{q}}-\mathbf{G}(\hat{\mathbf{p}})\mathbf{M}^{-1}\Big(\nabla\mathbf{g}(\hat{\mathbf{p}})+\mathbf{G}(\hat{\mathbf{p}})^T\bm{\lambda}\Big) \nonumber \\ 
=&0. \label{54}
\end{align}
Moreover, we denote 
\begin{align}
\mathbf{C}&=\mathbf{G}(\hat{\mathbf{p}})\mathbf{G}(\hat{\mathbf{p}})^T\nonumber\\
&=
\begin{bmatrix}
	\hat{\mathbf{p}}_1^T &  &  & \\
	&  \hat{\mathbf{p}}_2^T&  & \\
	&  & \ddots  & \\
	&  &  &\hat{\mathbf{p}}_L^T
\end{bmatrix} 
\begin{bmatrix}
	\hat{\mathbf{p}}_1 &  &  & \\
	&  \hat{\mathbf{p}}_2&  & \\
	&  & \ddots  & \\
	&  &  &\hat{\mathbf{p}}_L
	\end{bmatrix}\nonumber\\
&=
\begin{bmatrix}
	\hat{\mathbf{p}}_1^T\hat{\mathbf{p}}_1 &  &  & \\
	&  \hat{\mathbf{p}}_2^T\hat{\mathbf{p}}_2&  & \\
	&  & \ddots  & \\
	&  &  &\hat{\mathbf{p}}_L^T\hat{\mathbf{p}}_L
\end{bmatrix}\nonumber\\
	&=
\begin{bmatrix}
	\rho_1 &  &  & \\
	&  \rho_2&  & \\
	&  & \ddots  & \\
	&  &  &\rho_L
\end{bmatrix}\in\mathbb{R}^{L\times L} \label{C}
\end{align}
as the power constraint matrix. 
Then, equation \eqref{54} can be rewritten as
\begin{align}
	\mathbf{M}^{-2}\mathbf{Q}\hat{\mathbf{q}}-\mathbf{G}(\hat{\mathbf{p}})\mathbf{M}^{-1}\nabla\mathbf{g}(\hat{\mathbf{p}})-\mathbf{M}^{-1}\mathbf{C}\bm{\lambda}  =0.
\end{align}
Finally, by rearranging the above equation, we obtain the expression of \(\bm{\lambda}\) as
\begin{align}   \bm{\lambda}=\mathbf{C}^{-1}\Big(\mathbf{M}^{-1}\mathbf{Q}\hat{\mathbf{q}} -\mathbf{M}\mathbf{G}(\hat{\mathbf{p}})\mathbf{M}^{-1}\nabla \mathbf{g}(\hat{\mathbf{p}})\Big)\in \mathbb{R}^{L \times 1}.
\end{align}
	
\section{Proof for Theorem \ref{thm2}}
\label{fulu2}
From the constraint condition in \eqref{26a}, we have
\begin{align}
	 \mathbf{G}(\hat{\mathbf{p}}) \hat{\mathbf{q}}=\mathbf{0}.\label{fuluCon}
\end{align}
The equation \eqref{fuluCon} should be satisfied at each iteration in \eqref{diedaizong}.
Substituting this into the iterative update \eqref{diedai23} and left-multiplying both sides by \(\mathbf{G}(\hat{\mathbf{p}})\), we obtain
\begin{align}
&\mathbf{G}(\hat{\mathbf{p}}_{n+1}) \hat{\mathbf{q}}_{n+1}\nonumber\\
=&\mathbf{G}(\hat{\mathbf{p}}_{n+1})e^{-\gamma h/2}(\hat{\mathbf{q}}_{n+1/2}-\frac{h}{2}(\nabla \mathbf{g}(\hat{\mathbf{p}}_{n+1})+\mathbf{G}(\hat{\mathbf{p}}_{n+1})^T\bm{\mu}_n)\nonumber \\
=&\mathbf{0}.
\end{align}
Since \(e^{-\gamma h/2}\) is a scalar constant, it can be removed from both sides of the equation.
After rearranging the terms, we obtain the expression of $\bm{\mu}_n$ as 
\begin{align}
\bm{\mu}_n=\mathbf{C}^{-1}\frac{2\mathbf{G}(\hat{\mathbf{p}}_{n+1})\hat{\mathbf{q}}_{n+\frac12}-h\mathbf{G}(\hat{\mathbf{p}}_{n+1})\nabla \mathbf{g}(\hat{\mathbf{p}}_{n+1})}{h} \in \mathbb{R}^{L \times 1}
\end{align}
where matrix \(\mathbf{C}\) represents the power constraint matrix, and its expression is given in \eqref{C}.	
\bibliographystyle{IEEEtran}
\bibliography{IEEEabrv,reference}
\vspace{12pt}
	
\end{document}